\documentclass[twocolumn,superscriptaddress,aps]{revtex4-2}

\usepackage{graphicx}
\usepackage{amsmath}
\usepackage{amsthm}
\usepackage{amssymb}
\usepackage{latexsym}
\usepackage{array}
\usepackage{hyperref}
\usepackage{amsfonts}
\usepackage{dsfont}
\usepackage{mathrsfs}
\usepackage{verbatim}
\usepackage{bbold}
\usepackage[normalem]{ulem}
\usepackage{upgreek}
\usepackage{makecell}
\usepackage{adjustbox}
\usepackage{algorithm}
\usepackage{algpseudocode}
\usepackage{color}
\usepackage{bm}
\usepackage{times}
\usepackage{booktabs}
\usepackage{multirow}

\newcommand{\abs}[1]{\left|#1\right|}
\newcommand{\norm}[1]{\left\|#1\right\|}
\newcommand{\ket}[1]{\left|#1\right\rangle}
\newcommand{\bra}[1]{\left\langle #1\right|}
\newcommand{\braket}[2]{\left\langle #1|#2\right\rangle}
\newcommand{\ketbra}[2]{\ket{#1}\!\!\bra{#2}}
\newcommand{\Tr}{\operatorname{Tr}}
\newcommand{\End}{\operatorname{End}}
\newcommand{\id}{\hat{\mathds{1}}}
\newcommand{\R}{\mathbb{R}}
\newcommand{\C}{\mathbb{C}}

\newcommand{\cA}{\mathcal{A}}
\newcommand{\cD}{\mathcal{D}}

\newcommand{\SM}{Supplemental Material}
\newcommand{\bal}{\boxtimes}
\newcommand{\phys}{\mathrm{phys}}

\newcommand{\Sym}{\operatorname{Sym}}
\newcommand{\Span}{\operatorname{span}}

\newcommand{\Proj}{\mathbb P}

\theoremstyle{plain}
\newtheorem{theorem}{Theorem}
\newtheorem{proposition}{Proposition}

\newtheorem{corollary}{Corollary}
\theoremstyle{definition}
\newtheorem{definition}{Definition}
\newtheorem{remark}{Remark}

\makeatletter
\renewcommand\part[1]{%
  \clearpage
  \onecolumngrid
  \section*{#1}
  
}
\makeatother
\makeatletter
\let\origaddcontentsline\addcontentsline
\renewcommand{\addcontentsline}[3]{}
\makeatother
\begin{document}

\part{}

\title{Phase-Flag Access and No-Go Constraints on Quotient-Space Real Quantum Mechanics}

\author{Jeongho~Bang}\email{jbang@yonsei.ac.kr}
\affiliation{Institute for Convergence Research and Education in Advanced Technology, Yonsei University, Seoul 03722, Republic of Korea}
\affiliation{Department of Quantum Information, Yonsei University, Incheon 21983, Republic of Korea}

\author{Kyoungho~Cho}
\affiliation{Institute for Convergence Research and Education in Advanced Technology, Yonsei University, Seoul 03722, Republic of Korea}
\affiliation{Department of Statistics and Data Science, Yonsei University, Seoul 03722, Republic of Korea}

\author{Kyunghyun~Baek}
\affiliation{Institute for Convergence Research and Education in Advanced Technology, Yonsei University, Seoul 03722, Republic of Korea}
\affiliation{Department of Quantum Information, Yonsei University, Incheon 21983, Republic of Korea}

\date{\today}

\begin{abstract}
Barrios Hita \textit{et al.} [Phys. Rev. Lett. \textbf{136}, 240202 (2026)] proposed a quotient-space formulation of quantum mechanics over the real numbers and concluded that complex numbers are only a convenience. A preceding analysis [arXiv:2607.05865] showed that, when the construction is empirically equivalent to complex quantum mechanics, its real flag is protected by a hidden complex structure $\hat{J}$ and its composition rule is the balanced tensor product over that structure. Here we ask what would follow from the stronger operational reading that the flag is an accessible real degree of freedom. The answer is a direct clash with overlap-determinability (OD), introduced in [arXiv:2601.14638], and with standard no-go theorems. A readable or coherently controllable flag fixes a representative of each ray, and hence supplies precisely the phase convention that OD identifies as the missing resource for universal superposition of unknown states. Once promoted to a generic phase-access primitive, this resource enables probabilistic cloning of a linearly dependent set, steering-based signaling, and logarithmic-query unstructured search. Thus the quotient construction has a sharp follow-up interpretation: as a protected gauge representation it is standard complex quantum mechanics in real notation; as an accessible real theory it leaves the quantum operational framework.
\end{abstract}

\maketitle

%----------------------------------------------------------------------------------------------------------------------------------------------------------------------------------------------------------------
%\section{Introduction}
%----------------------------------------------------------------------------------------------------------------------------------------------------------------------------------------------------------------

{\em Introduction.}---There is a harmless and a non-harmless sense in which complex numbers can be removed from quantum mechanics. The harmless sense is representational: a vector $x+i y\in\C^d$ can be stored as the real pair $(x,y)\in\R^{2d}$. The non-harmless sense is operational: one asks whether the phase gauge, the observable algebra, and the tensor product over $\C$ can be replaced by a genuinely real operational theory. The latter question is what is tested by real-vs-complex Bell-network no-go theorems and experiments, where the foil theory uses real amplitudes and the ordinary tensor product over $\R$~\cite{Renou2021,Li2022,Chen2022,Wu2022,Ying2025}.

Barrios Hita \textit{et al.}~\cite{BarriosHita2026} proposed a different real-number formulation. Each complex state is embedded into a real vector equipped with a two-dimensional flag, and composite systems are obtained by quotienting out different distributions of the complex phase among subsystems. Because this quotient construction is designed to be one-to-one with complex quantum mechanics, it reproduces the complex predictions and evades the usual real-amplitude foil. The authors therefore argue that real-valued quantum mechanics cannot be experimentally falsified and that complex numbers are a matter of convenience.

Our previous work~\cite{BangChoBaek2026Hidden} addressed the structural part of this claim. It showed that empirical equivalence of the quotient construction requires a distinguished real linear complex structure $\hat{J}$, $\hat{J}^2=-\id$, and that physical effects, instruments, and dynamics must preserve the corresponding $SO(2)$ phase gauge. It also showed that the quotient composition rule is the balanced tensor product over this hidden $\hat{J}$, equivalently the complex tensor product written in real form. Hence the protected quotient theory is not the ordinary real-amplitude theory; it is standard complex quantum mechanics (CQM) expressed in real coordinates.

In this study, as a follow-up, we examine the remaining operational question: what if the phase flag is not merely a protected gauge coordinate but is treated as an accessible real degree of freedom? This question is forced by the rhetoric of a real-valued theory. If the flag is in the real state space but no operation can read or couple to it, then its inaccessibility is a superselection rule. If some operation can read or coherently control it, then global phase has become physical. We show that the latter reading is not merely ``different from CQM''; it collides with the overlap-determinability (OD) obstruction to universal superposition of unknown states~\cite{BangChoYee2026OD} and with familiar no-go theorems for cloning, no-signaling, and Grover search~\cite{Oszmaniec2016,Bandyopadhyay2020,DuanGuo1998,HJW1993,Bao2016,Grover1997,Zalka1999,BBBV1997,KumarParaoanu2011,Yee2020}.

The conceptual point is simple. In projective quantum mechanics an unknown pure input is a ray, not a vector with an absolute phase. In the theory of OD~\cite{BangChoYee2026OD}, it is true that a coherent expression such as $a\ket{\psi}+b\ket{\phi}$ is not a function of two unknown rays unless the physical scenario supplies a phase convention. The quotient flag stores exactly such a convention. Therefore, an accessible flag is an OD resource. The protected branch is safe precisely because it denies that resource.

%----------------------------------------------------------------------------------------------------------------------------------------------------------------------------------------------------------------
%\section{Structural input from Ref.~\cite{BangChoBaek2026Hidden}}
%----------------------------------------------------------------------------------------------------------------------------------------------------------------------------------------------------------------

{\em Structural input from Ref.~\cite{BangChoBaek2026Hidden}.}---Let $H\simeq\C^d$ and let $V=H_{\R}$ be its realification. In a fixed basis,
\begin{eqnarray}
\hat{S}\ket{\psi}=\operatorname{Re}\ket{\psi}\otimes\ket{0}_{F}+\operatorname{Im}\ket{\psi}\otimes\ket{1}_{F},
\label{eq:S_main}
\end{eqnarray}
and multiplication by $i$ is represented by a real skew-symmetric operator
\begin{eqnarray}
\hat{J}^2=-\id,
\quad
\hat{J}^T=-\hat{J},
\quad
\hat{S}(e^{i\alpha}\psi)=e^{\alpha \hat{J}}\hat{S}(\psi).
\label{eq:J_orbit}
\end{eqnarray}
Thus the complex global phase becomes an $SO(2)$ orbit in the flag plane.

The protected quotient construction identifies all points on this orbit. Consequently, any real effect $\hat{E}=\hat{E}^T$ that is physically meaningful on the quotient must obey
\begin{eqnarray}
v^T \hat{E} v=(e^{\alpha \hat{J}}v)^T \hat{E}(e^{\alpha \hat{J}}v)
\quad (\forall v,\alpha),
\label{eq:orbit_invariance}
\end{eqnarray}
which is equivalent to~\cite{BangChoBaek2026Hidden}
\begin{eqnarray}
[\hat{E},\hat{J}]=0.
\label{eq:commute_J}
\end{eqnarray}
The physical real observable algebra is therefore
\begin{eqnarray}
\cA_{\phys}=\Sym(V)\cap\{\hat{J}\}',
\label{eq:Aphys}
\end{eqnarray}
not the full cone of positive real effects. The same descent condition applies to instruments and dynamics. For example, a sufficient branchwise condition is
\begin{eqnarray}
\hat{\mathcal I}_r(e^{\alpha \hat{J}}\hat{\rho} e^{-\alpha \hat{J}})
=e^{\alpha \hat{J}}\hat{\mathcal I}_r(\hat{\rho})e^{-\alpha \hat{J}}
\quad (\forall r,\alpha).
\label{eq:inst_cov}
\end{eqnarray}
If this condition is dropped, outcome probabilities or posterior states depend on the chosen representative of a complex ray.

For composites, the quotient relation is the balanced tensor product over the hidden complex structures:
\begin{eqnarray}
V_A\bal V_B&=&(V_A\otimes_{\R}V_B)/N_{AB},\nonumber\\
N_{AB}&=&\Span_{\R}\{\hat{J}_A v\otimes w-v\otimes \hat{J}_B w\}.
\label{eq:balanced}
\end{eqnarray}
so that $\hat{J}_A v\bal w=v\bal \hat{J}_B w$. As shown in Ref.~\cite{BangChoBaek2026Hidden},
\begin{eqnarray}
V_A\bal V_B\simeq (H_A\otimes_{\C}H_B)_{\R}.
\label{eq:balanced_iso}
\end{eqnarray}
Eqs.~(\ref{eq:Aphys})--(\ref{eq:balanced_iso}) are the boundary conditions for the present paper. They define the branch in which the quotient theory is empirically equivalent to CQM. We now ask what happens outside that branch.

%----------------------------------------------------------------------------------------------------------------------------------------------------------------------------------------------------------------
%\section{Phase flags as OD resources}
%----------------------------------------------------------------------------------------------------------------------------------------------------------------------------------------------------------------

{\em Phase flags as OD resources.}---Let $\Proj(H)$ denote the pure-state ray space. A phase convention on a domain $\cD\subset\Proj(H)$ is a map $\Gamma$ assigning to each ray $\hat{P}\in\cD$ a normalized representative $\ket{\psi_\Gamma}$ with $\hat{P}=\ketbra{\psi_\Gamma}{\psi_\Gamma}$. The domain is overlap-determinable in a given operational scenario if that scenario supplies such a convention, so that $\braket{\psi_\Gamma}{\phi_\Gamma}$ is a definite complex datum for the relevant inputs. The OD theorem~\cite{BangChoYee2026OD} states, in particular, that a nonzero completely positive trace-nonincreasing map that produces a prescribed coherent superposition of unknown inputs on a domain exists if and only if the domain is OD in the corresponding scenario. The missing resource is not success probability; it is phase convention.

The same obstruction can be stated without invoking any dynamics. If one treats the real representatives themselves as ordinary vectors, the most tempting operation is simply to add them. But addition is not a function on the quotient unless a representative has first been chosen. This is the local form of the OD obstruction.

\begin{proposition}[Naive real superposition is representative-dependent]
\label{prop:naive_superposition}
Let the quotient identify $v\sim e^{\alpha \hat{J}}v$. For nonzero coefficients $a,b$, the rule
\begin{eqnarray}
([v],[w])\longmapsto [a v+b w]
\label{eq:naive_real_add}
\end{eqnarray}
is not a well-defined operation on pairs of generic complex rays. It becomes well defined only after an additional phase convention chooses representatives of the two orbits.
\end{proposition}

\begin{proof}[Proof sketch]
Take two real basis states represented by $v=\hat{S}\ket0$ and $w=\hat{S}\ket1$. The representatives $w$ and $\hat{J}w=\hat{S}(i\ket1)$ belong to the same complex ray. However, $a v+b w$ represents $a\ket0+b\ket1$, whereas $a v+b \hat{J}w$ represents $a\ket0+i b\ket1$. For $a b\neq0$ these are different rays. Thus the output depends on a gauge choice for the second input. A phase convention is precisely the extra datum that removes this ambiguity.
\end{proof}

The quotient flag is dangerous exactly because it carries this missing resource. If a $\hat{J}$-noncommuting operation distinguishes $v$ from $e^{\alpha \hat{J}}v$, it distinguishes representatives of a single complex ray. If, in addition, the flag can be coherently reset or controlled on arbitrary unknown inputs, it selects a representative. For real representatives $v,w\in V$, the complex overlap in the induced convention is recovered from
\begin{eqnarray}
\braket{\psi}{\phi}=v^T w-i v^T\hat{J}w.
\label{eq:overlap_from_real}
\end{eqnarray}
Thus, a coherently accessible flag converts projective input data into vector-representative data.

We use the following terminology to isolate the non-quantum assumption.

\begin{definition}[Strong phase access]
\label{def:strong_access}
A quotient-space real theory has strong phase access on a domain $\cD$ if the flag angle can be read, coherently aligned, or coherently controlled for arbitrary unknown inputs from $\cD$ in a way that permits convention-fixed interference operations. Equivalently, the theory supplies an operational phase convention $\Gamma_F$ on $\cD$.
\end{definition}

The adjective ``strong'' is essential. A single noisy $\hat{J}$-noncommuting imperfection does not automatically implement the protocols below. The point is conditional: if the flag is promoted from gauge coordinate to a generic phase reference, then the quotient construction supplies precisely the resource that OD says is absent in ordinary quantum mechanics.

\begin{theorem}[Flag access implies OD access]
\label{thm:flag_OD}
On any domain of unknown rays, strong phase access induces overlap-determinability. Consequently, a universal convention-fixed superposition primitive derived from the flag is not an operation of standard quantum mechanics unless the same domain was already OD by ordinary side information.
\end{theorem}

\begin{proof}[Proof sketch]
Strong access selects, for each ray $\hat{P}$, a representative $\ket{\psi_{\Gamma_F}}$ by aligning the flag to a fixed reference direction. This assignment is a phase convention by definition. Equation~(\ref{eq:overlap_from_real}) then makes the overlap between any two representatives a definite complex number. Hence the domain is OD. If the original CQM scenario supplied no such convention, the OD theorem forbids a nonzero CP trace-nonincreasing map that universally produces $a\ket{\psi}+b\ket{\phi}$ on that domain. Therefore the flag-derived primitive is an extra post-quantum resource, not a hidden quantum operation.
\end{proof}

This is the first sense in which the claim of Ref.~\cite{BarriosHita2026} must be read carefully. The protected quotient map may remove the symbol $i$ from coordinates, but it may not make the flag operational. The moment the flag becomes a phase reference for unknown states, it does exactly what the no-superposition theorem forbids in non-OD domains.

This observation gives a compact operational incompatibility. The following formulation separates the representational content of the quotient from the stronger real-operational reading.

\begin{theorem}[OD-postulate incompatibility]
\label{thm:OD_postulate}
For a generic domain of independently prepared unknown pure states, the following three requirements cannot be maintained simultaneously:
\begin{enumerate}
\item[(a)] physical pure states are quotient orbits $[v]=\{e^{\alpha \hat{J}}v\}$, as in CQM;
\item[(b)] the real representative space carries ordinary coherent operations that can add, compare, reflect about, or otherwise use the flag representatives of unknown inputs;
\item[(c)] the resulting theory obeys the OD/no-superposition constraint of standard quantum mechanics on non-OD domains.
\end{enumerate}
\end{theorem}

\begin{proof}[Proof sketch]
If (a) holds, no operation may depend on which $v\in[v]$ was chosen unless a phase convention has been supplied. If (b) includes ordinary coherent use of representatives, Proposition~\ref{prop:naive_superposition} shows that the operation either is not well defined on quotient states or implicitly fixes representatives. The latter is exactly OD access by Theorem~\ref{thm:flag_OD}. On a generic non-OD domain, OD access enables a universal superposition primitive that (c) forbids. Therefore at least one of the three requirements must be abandoned.
\end{proof}

There are therefore two layers of flag access. Weak access, such as a single noncommuting readout, is already enough to destroy empirical equivalence to CQM because it makes global phase observable or representative-dependent. Strong access is the additional assumption that this phase information can be fed coherently into transformations on unknown inputs. The no-go consequences below use the strong layer. This distinction prevents an overstatement: the protected quotient construction is not claimed to perform the following tasks; the point is that any interpretation that makes the flag operational enough to support a real-amplitude dynamics would perform them.

%----------------------------------------------------------------------------------------------------------------------------------------------------------------------------------------------------------------
%\section{No-go cascade from making the flag physical}
%----------------------------------------------------------------------------------------------------------------------------------------------------------------------------------------------------------------

{\em No-go cascade from making the flag physical.}---The OD conflict is not only formal. The same resource reproduces the standard cascade of impossible tasks.

\begin{theorem}[Phase-access no-go cascade]
\label{thm:cascade}
In a strong phase-access quotient theory on arbitrary unknown inputs, the following tasks become possible.
\begin{enumerate}
\item[(i)] A linearly dependent set of pure states can be probabilistically cloned.
\item[(ii)] The resulting discrimination/cloning primitive can be combined with steering to transmit information without classical communication.
\item[(iii)] If strong phase access supplies unit-cost reflections about the current unknown state, unstructured search with a unique marked item can be solved using $O(\log N)$ oracle queries.
\end{enumerate}
\end{theorem}

\begin{proof}[Proof sketch]
For (i), take
\begin{eqnarray}
\ket{\psi_1}=\ket{0},\quad
\ket{\psi_2}=\ket{1},\quad
\ket{\psi_3}=(\ket{0}-\ket{1})/\sqrt2,
\label{eq:dep_set}
\end{eqnarray}
which are linearly dependent. Append a fixed state $\ket{2}$ and use the phase convention supplied by the flag to implement, upon success, $\ket{\psi_j}\mapsto \alpha\ket{\psi_j}+\beta\ket{2}$ with $\alpha\beta\neq0$. In the basis $\{\ket{0},\ket{1},\ket{2}\}$ the three output columns have determinant $\alpha^2\beta$, hence are linearly independent. They can therefore be unambiguously discriminated with nonzero probability, and repreparation gives probabilistic cloning of the original dependent set. This contradicts the Duan-Guo boundary, according to which probabilistic cloning is possible exactly for linearly independent families~\cite{DuanGuo1998,Chefles1998}.

For (ii), the Hughston-Jozsa-Wootters theorem lets Alice remotely prepare different ensemble decompositions of the same reduced state on Bob's side~\cite{HJW1993}. If Bob has a device whose success statistics depend on the decomposition because it clones or discriminates the phase-accessible dependent set, repeated trials reveal Alice's choice. This is the standard steering route from forbidden state-dependent transformations to signaling. In flag language, the same danger appears even more directly: if branch phases or global flag angles are readable, phases of Kraus representatives that are gauge in CQM become signal-carrying data. The protected branch avoids this only because those phases remain $\hat{J}$-gauge.

For (iii), suppose that for the current, generally oracle-dependent state $\ket{\psi_r}$ one can implement $\hat{R}_{\psi_r}=2\ketbra{\psi_r}{\psi_r}-\id$ at unit cost. With the usual phase oracle $\hat{O}_f$ for a unique marked state $\ket{w}$, define $\ket{\psi_{r+1}}=\hat{R}_{\psi_r}\hat{O}_f\ket{\psi_r}$ and $a_r=\braket{w}{\psi_r}$. Then
\begin{eqnarray}
a_{r+1}=(3-4\abs{a_r}^2)a_r.
\label{eq:grover_rec}
\end{eqnarray}
Thus $\abs{a_r}^2$ grows by at least a factor four while it is at most $1/4$. Starting from the uniform state, $\abs{a_0}^2=1/N$, only $O(\log N)$ oracle calls are needed to reach constant success probability. Such unknown-state reflections are precisely OD resources, because $\hat{R}_\psi\ket{\phi}=2\braket{\psi}{\phi}\ket{\psi}-\ket{\phi}$ contains a convention-fixed overlap. They are unavailable in CQM and would collapse the usual Grover lower bound~\cite{Grover1997,Zalka1999,BBBV1997,KumarParaoanu2011,Yee2020}. Full protocols are given in the \SM.
\end{proof}

The theorem should not be misread. It does not say that the protected construction of Ref.~\cite{BarriosHita2026} already implements cloning, signaling, or logarithmic Grover search. It says the opposite: the construction remains quantum only because the flag is not an accessible phase reference. The no-go cascade is the operational cost of taking the flag as a real physical resource rather than as gauge.

%----------------------------------------------------------------------------------------------------------------------------------------------------------------------------------------------------------------
%\section{The operational dilemma for the quotient construction}
%----------------------------------------------------------------------------------------------------------------------------------------------------------------------------------------------------------------

{\em The operational dilemma for the quotient construction.}---We can now state the resulting interpretation.

\begin{theorem}[Quotient-space phase-flag dilemma]
\label{thm:dilemma}
A quotient-space real formulation based on the flag map $\hat{S}$ has two operationally distinct readings.
\begin{enumerate}
\item[(i)] In the protected reading, the $SO(2)$ flag orbit is gauge. Effects, instruments, dynamics, and composites preserve $\hat{J}$ as in Eqs.~(\ref{eq:Aphys})--(\ref{eq:balanced_iso}). The theory is empirically equivalent to CQM, but it is CQM in real notation rather than a generic real operational theory.
\item[(ii)] In the access reading, some flag operation is physical in a way that distinguishes or coherently controls representatives on the $SO(2)$ orbit. Then the theory either is not well defined on complex rays, or it splits a single complex ray into physically distinct states. If this access is strong on unknown domains, it is an OD resource and enters the post-quantum regime of Theorems~\ref{thm:flag_OD} and~\ref{thm:cascade}.
\end{enumerate}
\end{theorem}

\begin{proof}[Proof sketch]
Part (i) is exactly the protected boundary established in Ref.~\cite{BangChoBaek2026Hidden}: quotient probabilities and transformations must descend to the $SO(2)$ orbit, and composition must be balanced over $\hat{J}$. Part (ii) follows because a $\hat{J}$-noncommuting flag operation assigns different operational data to $v$ and $e^{\alpha \hat{J}}v$. If those are still declared to be one ray, probabilities depend on the chosen representative. If they are declared to be distinct states, complex global phase is observable. When this access can be used coherently on unknown inputs, it defines a phase convention and hence OD access; the no-go cascade follows.
\end{proof}

%----------------------------------------------------------------------------------------------------------------------------------------------------------------------------------------------------------------
%\section{Discussion}
%----------------------------------------------------------------------------------------------------------------------------------------------------------------------------------------------------------------

{\em Discussion.}---We have reformulated the quotient-space real representation of quantum mechanics as an operational phase-flag dilemma. The construction faithfully reproduces complex quantum mechanics only when the two-dimensional flag is protected as a gauge coordinate: all physical effects and instruments must commute with the hidden complex structure $\hat{J}$, and composite systems must be formed by the corresponding balanced tensor product. Once this protection is relaxed, the same flag becomes an operational phase reference that distinguishes representatives of a single complex ray. The resulting access branch is therefore not a harmless real rewriting of complex quantum mechanics, but a post-quantum extension whose strong coherent form enters the OD/no-go regime, enabling cloning-type transformations, steering-based signaling, and logarithmic-query search.

The main lesson of this study is that the quotient construction in Ref.~\cite{BarriosHita2026} has no third interpretation. It can be a protected gauge rewriting of CQM, in which case it is consistent but not an independent real-amplitude operational theory. Or it can be read as giving physical status to the flag, in which case it supplies exactly the phase convention that OD and the no-superposition theorem identify as unavailable for generic unknown states. The latter reading conflicts with the same principles that prohibit cloning of linearly dependent states, superluminal signaling from steering, and super-Grover search.

This also clarifies the relation to real-vs-complex network tests. Those experiments do not rule out a representation of CQM in real notation. They rule out the ordinary real-amplitude foil with the ordinary tensor product and unrestricted real operations. The quotient construction evades the tests by changing the composition rule and the observable algebra~\cite{BangChoBaek2026Hidden}. The present OD analysis adds that one cannot then re-interpret the hidden flag as an accessible real subsystem without leaving quantum theory.

Thus, complex coordinates may be dispensable, but complex phase structure is not. In the protected branch it appears as $\hat{J}$-superselection and balanced composition. In the access branch it becomes a phase oracle. The first is standard quantum mechanics in real clothes; the second is a post-quantum resource.

%----------------------------------------------------------------------------------------------------------------------------------------------------------------------------------------------------------------
{\em Note.}---The \SM~contains a compact restatement of the hidden-$\hat{J}$ facts used from Ref.~\cite{BangChoBaek2026Hidden}, the proof that strong flag access induces OD access, and explicit cloning, steering/signaling, and logarithmic-query search protocols.

%----------------------------------------------------------------------------------------------------------------------------------------------------------------------------------------------------------------
\begin{acknowledgments}
{\em Acknowledgments.}---This work was supported by the Ministry of Science, ICT and Future Planning (MSIP) by the National Research Foundation of Korea (RS-2024-00432214); the Institute of Information and Communications Technology Planning and Evaluation grant funded by the Korean government (RS-2019-II190003, ``Research and Development of Core Technologies for Programming, Running, Implementing and Validating of Fault-Tolerant Quantum Computing System''); the Ministry of Trade, Industry and Resources (MOTIR), Korea, under the project ``Industrial Technology Infrastructure Program'' (RS-2024-00466693); the Korean ARPA-H Project through the Korea Health Industry Development Institute (KHIDI), funded by the Ministry of Health \& Welfare, Republic of Korea (RS-2025-25456722). We acknowledge the Yonsei University Quantum Computing Project Group for providing support and access to the Quantum System One (Eagle Processor), which is operated at Yonsei University.
\end{acknowledgments}


%==============================================================================================================================
%==============================================================================================================================
%==============================================================================================================================

\clearpage
\newpage
\onecolumngrid

% Restore TOC entries for Supplementary only
\makeatletter
\let\addcontentsline\origaddcontentsline
\makeatother

\part{Supplemental Material for ``Phase-Flag Access and No-Go Constraints on Quotient-Space Real Quantum Mechanics''}

%\tableofcontents

%=====================================================================================================================
\section{Purpose and relation to arXiv:2607.05865}\label{sec:purpose_SM}
%=====================================================================================================================

This Supplemental Material gives the detailed proof and protocol layer behind the Letter. The Letter is written as a follow-up to Ref.~\cite{BangChoBaek2026Hidden}. That preceding work isolated the hidden-complex-structure content of quotient-space real quantum mechanics (RNQM): empirical equivalence to complex quantum mechanics (CQM) requires a distinguished real linear complex structure $\hat{J}$, a $\hat{J}$-preserving observable and dynamical algebra, and a balanced tensor product that is isomorphic to the complex tensor product. We therefore do not present the hidden-$\hat{J}$ reconstruction as the main new claim of the present paper. We restate the relevant facts here only to make the later no-go protocols self-contained.

The new focus is the operational status of the phase flag in the quotient construction of Ref.~\cite{BarriosHita2026}. If the flag is protected, it is a gauge coordinate and the construction is CQM in real notation. If the flag is treated as an accessible real degree of freedom, then it supplies an operational phase convention for unknown rays. In the terminology of the overlap-determinability (OD) analysis~\cite{BangChoYee2026OD}, this is precisely the missing resource that would make otherwise ill-defined coherent superpositions of unknown states into convention-fixed operations. Once promoted to a generic coherent primitive, such phase access yields the usual post-quantum consequences: probabilistic cloning of a linearly dependent family, steering-based signaling, and logarithmic-query Grover search.

We use the following notation. If $H$ is a complex Hilbert space, $H_{\R}$ denotes the same additive group regarded as a real vector space. Multiplication by $i$ is denoted by $\hat{J}$. Thus $\hat{J}^2=-\id$. The real inner product is $(v,w)_{\R}=\operatorname{Re}\braket{v}{w}_{\C}$. For a real space $V$, $\Sym(V)$ denotes the real symmetric operators on $V$. The symbol $\bal$ denotes the balanced tensor product over the hidden complex structures.

%=====================================================================================================================
\section{Realification, the flag, and the $\hat{J}$-commutant}\label{sec:realification}
%=====================================================================================================================

This section records, in explicit matrix form, the protected hidden-$\hat{J}$ structure used from Ref.~\cite{BangChoBaek2026Hidden}. These details are included so that the later access-branch protocols can be read without referring back to the previous paper.

Let $H\simeq \C^d$. In a fixed basis, write $\psi=x+i y$ with $x,y\in\R^d$. The realification $V=H_{\R}$ identifies $\psi$ with $(x,y)\in\R^d\oplus\R^d$. Multiplication by $i$ becomes
\begin{eqnarray}
\hat{J}(x,y)=(-y,x),
\quad
\hat{J}=\begin{pmatrix}0&-\id_d\\ \id_d&0\end{pmatrix},
\quad
\hat{J}^2=-\id_{2d}.
\label{eq:J_block_SM}
\end{eqnarray}
Moreover $\hat{J}^T=-\hat{J}$ and $\hat{J}^T \hat{J}=\id$.

The flag form used in Ref.~\cite{BarriosHita2026} is the same construction written as
\begin{eqnarray}
\hat{S}\ket{\psi}=\operatorname{Re}\ket{\psi}\otimes\ket{0}_F+\operatorname{Im}\ket{\psi}\otimes\ket{1}_F,
\label{eq:S_SM}
\end{eqnarray}
with
\begin{eqnarray}
\hat{J}_F\ket{0}_F=\ket{1}_F,
\quad
\hat{J}_F\ket{1}_F=-\ket{0}_F,
\quad
\hat{J}_F=\begin{pmatrix}0&-1\\1&0\end{pmatrix}.
\label{eq:JF_SM}
\end{eqnarray}
The global phase orbit becomes
\begin{eqnarray}
\hat{S}(e^{i\alpha}\psi)=e^{\alpha \hat{J}}\hat{S}(\psi).
\label{eq:phase_orbit_SM}
\end{eqnarray}

\begin{proposition}[Complex-linear maps]
\label{prop:complex_linear_SM}
A real linear map $\hat{A}_R\in\End_{\R}(V)$ is the realification of a complex-linear map $\hat{A}\in\End_{\C}(H)$ if and only if
\begin{eqnarray}
[\hat{A}_R,\hat{J}]=0.
\label{eq:complex_linear_comm_SM}
\end{eqnarray}
In block form this means
\begin{eqnarray}
\hat{A}_R=\begin{pmatrix}X&-Y\\Y&X\end{pmatrix},
\label{eq:block_comm_SM}
\end{eqnarray}
which represents the complex matrix $\hat{A}=X+iY$.
\end{proposition}

\begin{proof}---If $\hat{A}$ is complex linear, $\hat{A}(i\psi)=i\hat{A}\psi$, which is $\hat{A}_R\hat{J}=\hat{J}\hat{A}_R$. Conversely, write
\begin{eqnarray}
\hat{A}_R=\begin{pmatrix}P&Q\\R&S\end{pmatrix}.
\end{eqnarray}
The equation $\hat{A}_R\hat{J}=\hat{J}\hat{A}_R$ gives $Q=-R$ and $S=P$. Setting $X=P$ and $Y=R$ gives Eq.~(\ref{eq:block_comm_SM}).
\end{proof}

For Hermitian $\hat{A}=\hat{A}^{\dagger}$, the real representative is real symmetric and $\hat{J}$-commuting:
\begin{eqnarray}
\hat{T}(\hat{A})=\operatorname{Re}(\hat{A})\otimes \hat{I}_F+\operatorname{Im}(\hat{A})\otimes \hat{J}_F.
\label{eq:T_SM}
\end{eqnarray}
Thus the real observable algebra of complex quantum mechanics is not $\Sym(2d,\R)$, but
\begin{eqnarray}
\mathcal A_{\phys}=\Sym(V)\cap\{\hat{J}\}'.
\label{eq:phys_alg_SM}
\end{eqnarray}

%=====================================================================================================================
\section{$\hat{J}$-superselection and forbidden flag measurements}\label{sec:J_superselection}
%=====================================================================================================================

The physical pure state in complex quantum mechanics is a ray. In the real flag representation, this is the orbit
\begin{eqnarray}
v\sim e^{\alpha \hat{J}}v, \quad \alpha\in\R.
\label{eq:orbit_SM}
\end{eqnarray}
Operational equivalence means that all allowed effects assign identical probabilities to all representatives on the same orbit.

\begin{theorem}[$\hat{J}$-superselection]
\label{thm:J_super_SM}
Let $V$ be a finite-dimensional real Hilbert space with $\hat{J}^2=-\id$ and $\hat{J}^T=-\hat{J}$. Let $\hat{E}=\hat{E}^T$ be a real effect. The following are equivalent:
\begin{enumerate}
\item[(i)] $v^T \hat{E} v=(e^{\alpha \hat{J}}v)^T \hat{E}(e^{\alpha \hat{J}}v)$ for all $v\in V$ and $\alpha\in\R$.
\item[(ii)] $e^{-\alpha \hat{J}}\hat{E}e^{\alpha \hat{J}}=\hat{E}$ for all $\alpha\in\R$.
\item[(iii)] $[\hat{E},\hat{J}]=0$.
\end{enumerate}
\end{theorem}

\begin{proof}---Since $(e^{\alpha \hat{J}})^T=e^{-\alpha \hat{J}}$, condition (i) is equivalent to
\begin{eqnarray}
v^T\left(e^{-\alpha \hat{J}}\hat{E}e^{\alpha \hat{J}}-\hat{E}\right)v=0
\label{eq:quadratic_SM}
\end{eqnarray}
for all $v$. The matrix in parentheses is symmetric, hence Eq.~(\ref{eq:quadratic_SM}) for all $v$ implies (ii). Differentiating (ii) at $\alpha=0$ gives
\begin{eqnarray}
-\hat{J}\hat{E}+\hat{E}\hat{J}=0,
\end{eqnarray}
which is (iii). Conversely, if $[\hat{E},\hat{J}]=0$, then $\hat{E}$ commutes with $e^{\alpha \hat{J}}$ for all $\alpha$, so (ii) and hence (i) hold.
\end{proof}

\begin{corollary}[Flag-only effects are not physical in the protected branch]
\label{cor:flag_only_SM}
The projectors
\begin{eqnarray}
\hat{E}_0=\id_d\otimes\ketbra{0}{0}_F,
\quad
\hat{E}_1=\id_d\otimes\ketbra{1}{1}_F
\label{eq:E01_SM}
\end{eqnarray}
are not physical effects in a real representation of complex quantum mechanics.
\end{corollary}

\begin{proof}---Using Eq.~(\ref{eq:JF_SM}), $[\hat{E}_0,\hat{J}_F]\neq0$. Equivalently, let $x\in\R^d$ be nonzero. The two representatives $x\otimes\ket{0}_F$ and $\hat{J}(x\otimes\ket{0}_F)=x\otimes\ket{1}_F$ differ by the complex global phase $i$, but
\begin{eqnarray}
(x\otimes\ket{0})^T \hat{E}_0(x\otimes\ket{0})=\norm{x}^2,
\quad
(x\otimes\ket{1})^T \hat{E}_0(x\otimes\ket{1})=0.
\end{eqnarray}
Thus $\hat{E}_0$ distinguishes representatives of a single complex ray.
\end{proof}

\begin{theorem}[Postulate-level incompatibility]
\label{thm:postulate_incompat_SM}
For the flag representation $\hat{S}$, the following three assumptions are mutually incompatible:
\begin{enumerate}
\item[(a)] every real positive effect on the flag-extended real Hilbert space is physically admissible;
\item[(b)] the orbit $v\sim e^{\alpha \hat{J}}v$ represents one physical complex ray;
\item[(c)] the theory is empirically equivalent to complex quantum mechanics.
\end{enumerate}
\end{theorem}

\begin{proof}---Assume (a). Then the positive projector $\hat{E}_0=\id_d\otimes\ketbra{0}{0}_F$ is admissible. By Corollary~\ref{cor:flag_only_SM}, $\hat{E}_0$ assigns different probabilities to $v=x\otimes\ket{0}_F$ and $\hat{J}v=x\otimes\ket{1}_F$. These two vectors differ only by multiplication of the corresponding complex vector by the global phase $i$. Hence (a) contradicts (b). Since complex quantum mechanics cannot distinguish global-phase representatives, the same effect contradicts (c). Therefore at least one assumption must be rejected. The protected quotient construction rejects (a) by restricting the physical effects to the $\hat{J}$-commutant.
\end{proof}

The same restriction applies to POVMs effect by effect. It also applies to instruments. Let $\{\hat{\mathcal I}_r\}$ be a real instrument on density matrices. A sufficient covariance condition for compatibility with the gauge is
\begin{eqnarray}
\hat{\mathcal I}_r(e^{\alpha \hat{J}}\hat{\rho} e^{-\alpha \hat{J}})=e^{\alpha \hat{J}}\hat{\mathcal I}_r(\hat{\rho})e^{-\alpha \hat{J}} \quad (\forall r,\alpha).
\label{eq:instrument_cov_SM}
\end{eqnarray}
Without this condition, either the outcome probabilities or the posterior states leak flag-phase information.

%=====================================================================================================================
\section{Balanced tensor product and underdetermination of P4}\label{sec:balanced_P4}
%=====================================================================================================================

We also restate the composition fact from Ref.~\cite{BangChoBaek2026Hidden}. The quotient product is not the ordinary real tensor product. It is the real balanced product that implements complex scalar balancing.

\begin{definition}[Balanced real tensor product over $\hat{J}$]
Let $(V_A,\hat{J}_A)$ and $(V_B,\hat{J}_B)$ be real vector spaces with complex structures. Define
\begin{eqnarray}
V_A\bal V_B=\frac{V_A\otimes_{\R}V_B}{N_{AB}},
\quad
N_{AB}=\Span_{\R}\{\hat{J}_Av\otimes w-v\otimes \hat{J}_Bw\}.
\label{eq:balanced_def_SM}
\end{eqnarray}
The equivalence class of $v\otimes w$ is denoted $v\bal w$.
\end{definition}

The defining relation is
\begin{eqnarray}
\hat{J}_Av\bal w=v\bal \hat{J}_Bw.
\label{eq:balanced_relation_SM}
\end{eqnarray}
Applying it twice gives
\begin{eqnarray}
\hat{J}_Av\bal \hat{J}_Bw=v\bal \hat{J}_B^2w=-v\bal w.
\label{eq:double_J_SM}
\end{eqnarray}

\begin{theorem}[Balanced quotient equals complex tensor product]
\label{thm:balanced_iso_SM}
Let $H_A$ and $H_B$ be complex Hilbert spaces, with realifications $(V_A,\hat{J}_A)$ and $(V_B,\hat{J}_B)$. The map
\begin{eqnarray}
\Phi:V_A\bal V_B\to (H_A\otimes_{\C}H_B)_{\R},
\quad
\Phi(v\bal w)=v\otimes_{\C}w
\label{eq:Phi_SM}
\end{eqnarray}
is a canonical real-linear isomorphism.
\end{theorem}

\begin{proof}---The map $(v,w)\mapsto v\otimes_{\C}w$ is real bilinear and satisfies
\begin{eqnarray}
\Phi(\hat{J}_Av\otimes w-v\otimes \hat{J}_Bw)=iv\otimes_{\C}w-v\otimes_{\C}iw=0.
\end{eqnarray}
Hence it descends to the quotient. Its image contains all simple complex tensors and thus spans $H_A\otimes_{\C}H_B$ over $\C$ and over $\R$. Choosing complex bases $\{e_a\}$ and $\{f_b\}$, each pair $(a,b)$ contributes the two real classes $e_a\bal f_b$ and $\hat{J}_Ae_a\bal f_b$, so
\begin{eqnarray}
\dim_{\R}(V_A\bal V_B)=2\dim_{\C}H_A\dim_{\C}H_B=\dim_{\R}(H_A\otimes_{\C}H_B)_{\R}.
\end{eqnarray}
A surjective real-linear map between vector spaces of the same finite dimension is an isomorphism.
\end{proof}

For two flag systems, the balanced relation is equivalent to the kernel relations in Ref.~\cite{BarriosHita2026}:
\begin{eqnarray}
\ket{00}_{F_AF_B}+\ket{11}_{F_AF_B}\sim0,
\quad
\ket{01}_{F_AF_B}-\ket{10}_{F_AF_B}\sim0.
\label{eq:kernel_flags_SM}
\end{eqnarray}
Indeed, set $v=\ket{0}_{F_A}$ and $w=\ket{0}_{F_B}$ in $\hat{J}_Av\otimes w-v\otimes \hat{J}_Bw\sim0$ to obtain $\ket{10}-\ket{01}\sim0$. Set $w=\ket{1}_{F_B}$ to obtain $\ket{11}+\ket{00}\sim0$.

\begin{proposition}[P4 does not determine the quotient]
\label{prop:P4_SM}
The locality postulate that local operations act trivially on remote systems is satisfied both by the ordinary real tensor product and by the balanced quotient tensor product. Therefore it does not derive the quotient rule.
\end{proposition}

\begin{proof}---For the ordinary real tensor product choose
\begin{eqnarray}
\xi^{\otimes}_{\R}(v,w)=v\otimes_{\R}w,
\quad
\eta^{1,\otimes}_{\R}(\hat{A})=\hat{A}\otimes\id,
\quad
\eta^{2,\otimes}_{\R}(\hat{B})=\id\otimes \hat{B}.
\end{eqnarray}
Then $\eta^{1,\otimes}_{\R}(\hat{A})\xi^{\otimes}_{\R}(v,w)=\xi^{\otimes}_{\R}(\hat{A}v,w)$, similarly for subsystem $2$, and the two local embeddings commute.

For the quotient choose
\begin{eqnarray}
\xi^{\bal}_{\R}(v,w)=v\bal w,
\quad
\eta^{1,\bal}_{\R}(\hat{A})(v\bal w)=\hat{A}v\bal w,
\quad
\eta^{2,\bal}_{\R}(\hat{B})(v\bal w)=v\bal \hat{B}w,
\end{eqnarray}
for $\hat{J}$-compatible local operators. These maps again satisfy the same local-triviality equations. Since two inequivalent composition rules satisfy the same local condition, the condition cannot select one of them.
\end{proof}

%=====================================================================================================================
\section{Formal proof of the phase-flag dilemma}\label{sec:dilemma_proof}
%=====================================================================================================================

\begin{definition}[Protected branch]
A real flag theory is in the protected branch if all physical effects satisfy $[\hat{E},\hat{J}]=0$, all instruments descend to the orbit quotient $v\sim e^{\alpha \hat{J}}v$, and composites are formed by the balanced tensor product $\bal$ over the chosen complex structures.
\end{definition}

\begin{definition}[Access branch]
A real flag theory is in the access branch if some physical effect, instrument, or dynamical coupling does not commute with the relevant $\hat{J}$ and therefore can distinguish, control, or decohere points on the orbit $v\sim e^{\alpha \hat{J}}v$.
\end{definition}

\begin{theorem}[Phase-flag dilemma]
\label{thm:dilemma_SM}
A quotient-space real formulation with the flag map $\hat{S}$ has the following alternatives.
\begin{enumerate}
\item In the protected branch, the physical algebra and composition are isomorphic to those of complex quantum mechanics. The flag is a gauge coordinate, not an ordinary subsystem.
\item In the access branch, either probabilities are not well defined on complex rays, or the phase orbit is split into physically distinct real states. In the latter case the theory is not empirically equivalent to complex quantum mechanics.
\end{enumerate}
\end{theorem}

\begin{proof}---If complex rays are to remain physical states, Eq.~(\ref{eq:orbit_SM}) is an equivalence relation. By Theorem~\ref{thm:J_super_SM}, every physical effect must commute with $\hat{J}$. The same descent condition for instruments gives Eq.~(\ref{eq:instrument_cov_SM}), and Theorem~\ref{thm:balanced_iso_SM} shows that composition must be balanced over $\hat{J}$ in order to reproduce complex tensor products. These conditions define exactly the protected branch. Proposition~\ref{prop:complex_linear_SM} and Theorem~\ref{thm:balanced_iso_SM} then give an isomorphism to the realification of complex quantum mechanics.

If a $\hat{J}$-noncommuting operation is allowed while one still identifies $v\sim e^{\alpha \hat{J}}v$, then Theorem~\ref{thm:J_super_SM} implies that probabilities depend on which representative is chosen, so the theory is not well defined on rays. If one instead abandons the identification, then $v$ and $e^{\alpha \hat{J}}v$ become distinct physical states. In particular, Corollary~\ref{cor:flag_only_SM} shows that $\hat{E}_0$ distinguishes $v$ from $\hat{J}v$, i.e., it makes complex global phase observable. This contradicts empirical equivalence to complex quantum mechanics.
\end{proof}

\begin{remark}
The theorem is deliberately not phrased as ``the PRL construction violates no-signaling.'' The PRL construction, read in the protected branch, does not. The point is that its empirical equivalence relies on a hidden $\hat{J}$-superselection rule. The no-go protocols below concern the access branch, especially the strengthened case where flag access becomes a generic phase-reference primitive.
\end{remark}

%=====================================================================================================================
\section{Phase-access RNQM as an OD resource}\label{sec:phase_access_OD}
%=====================================================================================================================

In projective quantum mechanics the input data of a pure state is a ray $\hat{P}_\psi=\ketbra{\psi}{\psi}$, not a vector representative. A coherent expression such as $\alpha\ket{\psi}+\beta\ket{\phi}$ is not a function of the two rays unless a relative phase convention has been supplied.

\begin{definition}[Operational phase convention and OD]
Let $\mathcal R\subseteq\mathbb P(H)$ be a set of rays. A phase convention on $\mathcal R$ is a map $\Gamma$ assigning to each ray $\hat{P}_\psi\in\mathcal R$ a normalized representative $\ket{\psi_\Gamma}$ with $\hat{P}_\psi=\ketbra{\psi_\Gamma}{\psi_\Gamma}$. The domain is overlap-determinable (OD) in a physical scenario if that scenario supplies such a convention operationally, so that $\braket{\psi_\Gamma}{\phi_\Gamma}$ is a definite complex datum for all relevant rays.
\end{definition}

\subsection{Real representatives and convention-fixed overlaps}
\label{subsec:real_overlap_SM}

Let $v=(x,y)$ and $w=(u,z)$ be the real representatives of $\psi=x+i y$ and $\phi=u+i z$. With the physics convention that $\braket{\psi}{\phi}$ is linear in the second slot,
\begin{eqnarray}
\braket{\psi}{\phi}=x\cdot u+y\cdot z+i(x\cdot z-y\cdot u).
\label{eq:complex_overlap_components_SM}
\end{eqnarray}
The corresponding real quantities are
\begin{eqnarray}
v^T w=x\cdot u+y\cdot z,
\quad
v^T \hat{J}w=-x\cdot z+y\cdot u.
\label{eq:real_overlap_components_SM}
\end{eqnarray}
Thus
\begin{eqnarray}
\braket{\psi}{\phi}=v^T w-i v^T \hat{J}w.
\label{eq:overlap_formula_SM}
\end{eqnarray}
Equation~(\ref{eq:overlap_formula_SM}) is only a change of notation in the protected branch, because no operation selects a representative on the orbit $v\sim e^{\alpha \hat{J}}v$. It becomes an operational resource only when the flag can be read or coherently aligned for unknown inputs.

A $\hat{J}$-noncommuting flag effect is an operation that distinguishes different points on the orbit $\hat{S}\ket{\psi}\mapsto e^{\alpha \hat{J}}\hat{S}\ket{\psi}$. If this access is available uniformly, it can be used to calibrate a representative of each input ray relative to a fixed flag direction. In that sense, the access branch supplies the phase convention that standard quantum mechanics lacks. This is why the protocols below are stated for a strengthened access theory.

\begin{definition}[Strong phase-access RNQM]
\label{def:strong_access}
A strong phase-access RNQM is an access-branch real theory with the following operational primitive. For arbitrary unknown pure inputs from a tested domain, the flag sector can be read and coherently controlled so as to implement convention-fixed interference operations. Equivalently, the theory supplies an OD convention on that domain.
\end{definition}

The adjective ``strong'' is important. A single noncommuting noisy effect does not automatically implement all the protocols below. The protocols characterize the consequence of promoting the flag angle to a generic operational resource.

\begin{theorem}[Strong flag access induces OD]
\label{thm:flag_access_OD_SM}
Strong phase-access RNQM on a ray domain $\mathcal R$ induces an operational phase convention on $\mathcal R$. Hence $\mathcal R$ is overlap-determinable in the corresponding scenario.
\end{theorem}

\begin{proof}---For each ray $\hat{P}\in\mathcal R$, choose any real representative $v$ on its flag orbit. Strong phase access aligns the flag to a fixed reference direction, producing a representative $v_{\Gamma_F}$ whose residual ambiguity is common to the task and therefore fixed by the same reference. Let $\ket{\psi_{\Gamma_F}}$ be the complex vector represented by $v_{\Gamma_F}$. Then $\hat{P}=\ketbra{\psi_{\Gamma_F}}{\psi_{\Gamma_F}}$, so $\Gamma_F:\hat{P}\mapsto\ket{\psi_{\Gamma_F}}$ is a phase convention. For any two rays in the domain, Eq.~(\ref{eq:overlap_formula_SM}) gives a definite complex overlap between the selected representatives. Therefore the domain is OD.
\end{proof}

\begin{corollary}[No-superposition clash]
If a domain is not OD in ordinary CQM, then a universal superposition operation on that domain derived from strong flag access is not an ordinary quantum operation; it is an additional post-quantum primitive.
\end{corollary}

\begin{proof}---This is the contrapositive of the OD theorem of Ref.~\cite{BangChoYee2026OD}. A nonzero completely positive trace-nonincreasing universal superposer requires OD. If ordinary CQM supplies no OD convention but the flag-access theory does, then the flag has added a new operational resource.
\end{proof}

This distinction is the logical firewall of the paper. The protected branch of Ref.~\cite{BarriosHita2026}, read literally as a quotient/gauge theory, is not claimed to implement the protocols below. The protocols are counterfactual stress tests of the access branch: they show what becomes possible if the same flag that stores the hidden complex phase is promoted from a gauge coordinate to a coherent physical reference that works on arbitrary unknown inputs.

%=====================================================================================================================
\section{Protocol I: probabilistic cloning from phase access}\label{sec:cloning}
%=====================================================================================================================

The standard Duan-Guo theorem states that a finite set of pure states can be probabilistically cloned if and only if the vectors are linearly independent~\cite{DuanGuo1998}. Strong phase access violates this boundary by enabling a superposition primitive on unknown rays. We give an explicit three-state protocol.

\subsection{Primitive}

Assume a phase-access superposer $\hat{\mathsf S}_{\alpha,\beta}$ with nonzero $\alpha,\beta$ that, upon success, maps
\begin{eqnarray}
\ket{\psi}\ket{\phi}\longmapsto
\frac{\alpha\ket{\psi}+\beta\ket{\phi}}{\norm{\alpha\ket{\psi}+\beta\ket{\phi}}}
\label{eq:superposer_primitive_SM}
\end{eqnarray}
in the OD convention supplied by the accessible flag. This is precisely the kind of coherent operation forbidden for generic unknown rays in standard quantum mechanics~\cite{Oszmaniec2016,Bandyopadhyay2020,BangChoYee2026OD}.

\subsection{Dependent inputs and independent outputs}

Let $H=\Span\{\ket{0},\ket{1},\ket{2}\}$. Define the linearly dependent set
\begin{eqnarray}
\ket{\psi_1}=\ket{0},
\quad
\ket{\psi_2}=\ket{1},
\quad
\ket{\psi_3}=\frac{\ket{0}-\ket{1}}{\sqrt2},
\label{eq:dependent_set_SM}
\end{eqnarray}
and let the fixed auxiliary state be
\begin{eqnarray}
\ket{\phi}=\ket{2}.
\end{eqnarray}
The three states in Eq.~(\ref{eq:dependent_set_SM}) lie in a two-dimensional subspace, hence are linearly dependent.

On input $\ket{\psi_j}\ket{\phi}$, the primitive outputs
\begin{eqnarray}
\ket{\Psi_j}=\alpha\ket{\psi_j}+\beta\ket{2}
\label{eq:Psi_j_SM}
\end{eqnarray}
up to normalization. The coefficient matrix of $\{\ket{\Psi_j}\}$ in the basis $\{\ket{0},\ket{1},\ket{2}\}$ is
\begin{eqnarray}
M=\begin{pmatrix}
\alpha&0&\alpha/\sqrt2\\
0&\alpha&-\alpha/\sqrt2\\
\beta&\beta&\beta
\end{pmatrix}.
\label{eq:matrix_ind_SM}
\end{eqnarray}
Its determinant is
\begin{eqnarray}
\det M=\alpha^2\beta.
\label{eq:det_SM}
\end{eqnarray}
Thus $\{\ket{\Psi_1},\ket{\Psi_2},\ket{\Psi_3}\}$ is linearly independent whenever $\alpha\beta\neq0$.

%=====================================================================================================================
\subsection{Unambiguous discrimination and cloning}

A finite family of pure states is unambiguously distinguishable with nonzero success probability if and only if it is linearly independent~\cite{Chefles1998}. Therefore there exists a POVM $\{\hat{D}_1,\hat{D}_2,\hat{D}_3,\hat{D}_?\}$ such that
\begin{eqnarray}
\bra{\Psi_j}\hat{D}_i\ket{\Psi_j}=0\quad (i\neq j),
\quad
\bra{\Psi_j}\hat{D}_j\ket{\Psi_j}>0.
\label{eq:USD_SM}
\end{eqnarray}
The probabilistic cloning protocol is now:

\begin{center}
\fbox{\begin{minipage}{0.93\textwidth}
\textbf{Protocol 1: Probabilistic cloning of a linearly dependent set in strong phase-access RNQM.}
\begin{algorithmic}[1]
\Require One unknown input promised to be one of $\{\ket{\psi_1},\ket{\psi_2},\ket{\psi_3}\}$.
\State Append the fixed auxiliary state $\ket{\phi}=\ket{2}$.
\State Apply the phase-access superposer $\hat{\mathsf S}_{\alpha,\beta}$.
\State On success, perform the USD POVM $\{\hat{D}_1,\hat{D}_2,\hat{D}_3,\hat{D}_?\}$ on $\ket{\Psi_j}$.
\If{the conclusive outcome is $j$}
\State Prepare $\ket{\psi_j}\ket{\psi_j}$.
\Else
\State Declare failure.
\EndIf
\end{algorithmic}
\end{minipage}}
\end{center}

The physical meaning of this protocol is simple. The phase-access superposer does not merely append an auxiliary state; it creates a coherent embedding whose off-diagonal terms depend on a fixed phase convention. Ordinary quantum mechanics forbids such a universal embedding for unknown rays because it would convert projective input data into vector-representative data. Strong phase access assumes exactly this missing conversion.

For input $\ket{\psi_j}$, the total success probability is the product of the nonzero superposer success probability and the nonzero USD success probability for $\ket{\Psi_j}$. Hence the protocol probabilistically clones a linearly dependent set, which is impossible in ordinary quantum mechanics.

%=====================================================================================================================
\section{Protocol II: signaling from phase access}\label{sec:signaling}
%=====================================================================================================================

We give two signaling constructions. The first is a direct flag-phase signaling protocol tailored to the phase-access branch. The second is the standard steering-to-forbidden-cloning route. This is consistent with the broader equivalence between signaling resources and super-Grover resources in several post-quantum models~\cite{Bao2016}.

%=====================================================================================================================
\subsection{Direct flag-phase signaling by Kraus-phase choices}\label{subsec:direct_signal}

This protocol shows why a physical global phase is incompatible with no-signaling in the presence of ordinary entanglement and local measurement phases.

Alice and Bob share the Bell state
\begin{eqnarray}
\ket{\Phi^+}_{AB}=\frac{\ket{00}+\ket{11}}{\sqrt2}.
\label{eq:bell_signal_SM}
\end{eqnarray}
Alice measures in the $X$ basis $\{\ket{+},\ket{-}\}$, but she chooses one of two phase-equivalent Kraus implementations indexed by a bit $b\in\{0,1\}$:
\begin{eqnarray}
\hat{K}^{(b)}_{\pm}=e^{-ib\pi/2}\ket{m_{\pm}}\bra{\pm}.
\label{eq:kraus_phase_SM}
\end{eqnarray}
The effects are independent of $b$:
\begin{eqnarray}
\hat{K}^{(b)\dagger}_{\pm}\hat{K}^{(b)}_{\pm}=\ketbra{\pm}{\pm}.
\end{eqnarray}
In complex quantum mechanics these are the same local measurement instrument up to a phase of each Kraus branch; Bob's reduced density matrix is $\id_2/2$ for both choices.

However, the unnormalized conditional vector on Bob's side acquires the corresponding global phase:
\begin{eqnarray}
\bra{\pm}\Phi^+\rangle=\frac{1}{\sqrt2}\ket{\pm},
\quad
\hat{K}^{(b)}_{\pm}~\Rightarrow~e^{ib\pi/2}\ket{\pm}_B.
\label{eq:bob_phase_SM}
\end{eqnarray}
In the protected branch these two conditional states represent the same ray. In the access branch they are distinct flag states. Since $\ket{\pm}$ are real vectors,
\begin{eqnarray}
\hat{S}\ket{\pm}=\ket{\pm}\ket{0}_F,
\quad
\hat{S}(i\ket{\pm})=\ket{\pm}\ket{1}_F.
\label{eq:real_pm_SM}
\end{eqnarray}
Thus Bob's unconditional real ensemble is
\begin{eqnarray}
\hat{\rho}_B^{(0)}=\frac{\id_2}{2}\otimes\ketbra{0}{0}_F, 
\quad
\hat{\rho}_B^{(1)}=\frac{\id_2}{2}\otimes\ketbra{1}{1}_F.
\label{eq:rho_b_signal_SM}
\end{eqnarray}
If Bob is allowed to measure $\hat{E}_0=\id_2\otimes\ketbra{0}{0}_F$, then
\begin{eqnarray}
\Tr(\hat{E}_0\hat{\rho}_B^{(0)})=1,
\quad
\Tr(\hat{E}_0\hat{\rho}_B^{(1)})=0.
\label{eq:signal_prob_SM}
\end{eqnarray}
The operational steps are summarized as follows:
\begin{center}
\fbox{\begin{minipage}{0.93\textwidth}
\textbf{Protocol 2: Direct flag-phase signaling in the access branch.}
\begin{algorithmic}[1]
\Require A shared Bell pair $\ket{\Phi^+}_{AB}$ and access to Bob's flag effect $\hat{E}_0$.
\State Alice chooses a bit $b\in\{0,1\}$.
\State Alice implements the phase-equivalent $X$-basis instrument with Kraus phases $\hat{K}^{(b)}_\pm$.
\State Alice does not communicate the outcome $\pm$.
\State Bob measures $\hat{E}_0=\id_2\otimes\ketbra{0}{0}_F$.
\State The statistics distinguish $b=0$ from $b=1$ by Eq.~(\ref{eq:signal_prob_SM}).
\end{algorithmic}
\end{minipage}}
\end{center}
Bob learns Alice's bit without classical communication. Therefore no-signaling forces the branch phases in Eq.~(\ref{eq:kraus_phase_SM}) to remain gauge, which is exactly the $\hat{J}$-superselection rule.

%=====================================================================================================================
\subsection{Steering plus forbidden cloning}\label{subsec:steering_signal}

Protocol I also leads to signaling by the standard HJW steering argument~\cite{HJW1993}. Suppose Bob has a phase-access device that probabilistically clones a linearly dependent set $S=\{\ket{\psi_j}\}$ with success probabilities $p_j>0$. Let
\begin{eqnarray}
\hat{\rho}_B=\sum_j q_j\ketbra{\psi_j}{\psi_j}
\label{eq:rho_ensemble_S_SM}
\end{eqnarray}
be a mixed state with a decomposition into states from $S$. The HJW theorem says that any ensemble decomposition of $\hat{\rho}_B$ can be prepared on Bob's side by an appropriate measurement on a purifying system held by Alice.

Alice encodes a bit by choosing between two measurements that steer Bob to two decompositions of the same $\hat{\rho}_B$:
\begin{eqnarray}
\mathcal E_0=\{q_j,\ket{\psi_j}\},
\quad
\mathcal E_1=\{r_k,\ket{\phi_k}\}.
\label{eq:two_ensembles_SM}
\end{eqnarray}
Bob applies the forbidden cloner and records its success flag. For Alice's first choice the success probability is
\begin{eqnarray}
P_{\mathrm{succ}|0}=\sum_j q_jp_j>0.
\label{eq:Psucc0_SM}
\end{eqnarray}
For a generic second decomposition the average success probability $P_{\mathrm{succ}|1}=\sum_k r_k p(\phi_k)$ differs from Eq.~(\ref{eq:Psucc0_SM}). If the success functional were the same for every decomposition of every density matrix, it would be affine in $\hat{\rho}$ and would not support the nonlinear state-dependent cloning action on the dependent set. Thus a valid choice of $\mathcal E_1$ with $P_{\mathrm{succ}|1}\neq P_{\mathrm{succ}|0}$ exists. Repeating the experiment over many shared pairs lets Bob distinguish Alice's choice with arbitrarily high confidence.

The direct flag protocol above is stronger and more transparent for the present context: once global phase is readable, even Kraus phases that are operationally invisible in complex quantum mechanics become signal-carrying.

%=====================================================================================================================
\section{Protocol III: logarithmic-query Grover search}\label{sec:grover}
%=====================================================================================================================

The OD manuscript emphasizes that access to convention-fixed overlaps enables reflections about unknown states. Standard Grover search uses $\Theta(\sqrt N)$ queries and is optimal in ordinary quantum mechanics~\cite{Grover1997,Zalka1999,BBBV1997}; unknown-state reflections are known to collapse this geometry when treated as primitives~\cite{KumarParaoanu2011,Yee2020}. In the phase-access branch, the corresponding primitive is
\begin{eqnarray}
\hat{R}_{\psi}=2\ketbra{\psi}{\psi}-\id,
\label{eq:unknown_reflection_SM}
\end{eqnarray}
where $\ket{\psi}$ is not classically specified but is fixed by the phase-access convention. This operation is not available at unit cost in ordinary quantum mechanics for an arbitrary unknown $\ket{\psi}$.

Consider unstructured search over $N=2^n$ items, whose ordinary quantum query complexity is $\Theta(\sqrt N)$~\cite{Grover1997,Zalka1999}, with a unique marked state $\ket{w}$ and phase oracle
\begin{eqnarray}
\hat{O}_f\ket{x}=(-1)^{f(x)}\ket{x}, \quad f(w)=1,
\label{eq:oracle_SM}
\end{eqnarray}
where $f(x)=0$ for $x\neq w$. Let
\begin{eqnarray}
\ket{\psi_0}=\frac{1}{\sqrt N}\sum_{x=0}^{N-1}\ket{x}.
\label{eq:uniform_SM}
\end{eqnarray}
The phase-access Grover iteration is
\begin{eqnarray}
\ket{\psi_{r+1}}=\hat{R}_{\psi_r}\hat{O}_f\ket{\psi_r}.
\label{eq:super_grover_iter_SM}
\end{eqnarray}
Each round uses one oracle query and one unknown-state reflection about the current state.

\begin{center}
\fbox{\begin{minipage}{0.93\textwidth}
\textbf{Protocol 3: Phase-access logarithmic Grover search.}
\begin{algorithmic}[1]
\Require Phase oracle $\hat{O}_f$ with a unique marked item $w$; strong phase-access reflection primitive $\hat{R}_{\psi}$.
\State Prepare $\ket{\psi_0}=N^{-1/2}\sum_x\ket{x}$.
\For{$r=0,1,\ldots, r_*-1$}
\State Query the oracle: $\ket{\eta_r}=\hat{O}_f\ket{\psi_r}$.
\State Reflect about the current unknown state: $\ket{\psi_{r+1}}=\hat{R}_{\psi_r}\ket{\eta_r}$.
\EndFor
\State Measure $\ket{\psi_{r_*}}$ in the computational basis.
\end{algorithmic}
\end{minipage}}
\end{center}

Let
\begin{eqnarray}
a_r=\braket{w}{\psi_r},
\quad
p_r=\abs{a_r}^2.
\label{eq:ar_pr_SM}
\end{eqnarray}
Since
\begin{eqnarray}
\bra{\psi_r}\hat{O}_f\ket{\psi_r}=1-2p_r,
\quad
\bra{w}\hat{O}_f\ket{\psi_r}=-a_r,
\label{eq:oracle_expect_SM}
\end{eqnarray}
we obtain
\begin{eqnarray}
a_{r+1}
&=&\bra{w}\hat{R}_{\psi_r}\hat{O}_f\ket{\psi_r}\\
&=&2\braket{w}{\psi_r}\bra{\psi_r}\hat{O}_f\ket{\psi_r}-\bra{w}\hat{O}_f\ket{\psi_r}\\
&=&2a_r(1-2p_r)+a_r\\
&=&(3-4p_r)a_r.
\label{eq:ar_recursion_SM}
\end{eqnarray}
Therefore
\begin{eqnarray}
p_{r+1}=p_r(3-4p_r)^2.
\label{eq:pr_recursion_SM}
\end{eqnarray}
As long as $p_r\le1/4$, we have $3-4p_r\ge2$, so
\begin{eqnarray}
p_{r+1}\ge4p_r.
\label{eq:four_growth_SM}
\end{eqnarray}
Since $p_0=1/N$, after
\begin{eqnarray}
r_* = \left\lceil \log_4\left(\frac{N}{4}\right)\right\rceil
=\Theta(\log N)
\label{eq:rstar_SM}
\end{eqnarray}
iterations the success probability is at least $1/4$. A constant number of repetitions boosts the success probability to a constant close to one. Thus the $\Theta(\sqrt N)$ Grover bound collapses to $O(\log N)$ oracle queries in a model where $\hat{R}_{\psi_r}$ is available at unit cost.

This is not a standard Grover iterate with a fixed diffusion operator. The reflection axis is updated to the unknown, oracle-dependent state currently held by the computer. That update is nonlinear as a transformation on the projective state trajectory: the next operation depends on the present state itself. The logarithmic growth in Eq.~(\ref{eq:four_growth_SM}) is therefore a diagnostic of the extra resource, not a hidden quantum algorithm inside ordinary unitary mechanics.

The connection to phase access is visible from the action of $\hat{R}_{\psi}$ on a test state $\ket{\phi}$:
\begin{eqnarray}
\hat{R}_{\psi}\ket{\phi}=2\braket{\psi}{\phi}\ket{\psi}-\ket{\phi}.
\label{eq:reflection_overlap_SM}
\end{eqnarray}
For an unknown or oracle-dependent $\ket{\psi}$, the complex number $\braket{\psi}{\phi}$ is a convention-fixed overlap. Thus a unit-cost unknown-state reflection is an OD resource in operational form. It is precisely what the protected branch with $[\hat{E},\hat{J}]=0$ denies.

%=====================================================================================================================
\section{Theoremized protocol consequences of strong phase access}\label{sec:theoremized_protocols}
%=====================================================================================================================

This section packages the constructive protocols of Protocols I--III as theorems. The phrase ``strong phase access'' always refers to Definition~\ref{def:strong_access}; it means that the flag does more than suffer a local imperfection. It supplies a coherent OD convention on the tested domain.

\begin{theorem}[Cloning consequence]
\label{thm:cloning_summary_SM}
In a strong phase-access RNQM with a nonzero coherent superposition primitive $\hat{\mathsf S}_{\alpha,\beta}$, $\alpha\beta\neq0$, there exists a probabilistic protocol that clones a linearly dependent set of pure states.
\end{theorem}

\begin{proof}---Use the three dependent inputs in Eq.~(\ref{eq:dependent_set_SM}) and append the fixed state $\ket{2}$. The phase-access superposer produces the three states in Eq.~(\ref{eq:Psi_j_SM}). Their coefficient matrix is Eq.~(\ref{eq:matrix_ind_SM}) and has determinant $\det M=\alpha^2\beta\neq0$. Hence the outputs are linearly independent. A linearly independent finite family is unambiguously distinguishable with nonzero success probability. After a conclusive outcome $j$, one can reprepare $\ket{\psi_j}\ket{\psi_j}$. The total success probability is the product of the nonzero superposer success probability and the nonzero unambiguous-discrimination probability. Thus a linearly dependent input family is probabilistically cloned, which is forbidden in ordinary quantum mechanics.
\end{proof}

\begin{theorem}[Signaling consequence]
\label{thm:signaling_summary_SM}
If the primitive of Theorem~\ref{thm:cloning_summary_SM} is available, then a steering protocol violates no-signaling.
\end{theorem}

\begin{proof}---Let Bob's reduced state $\hat{\rho}_B$ have two ensemble decompositions as in Eq.~(\ref{eq:two_ensembles_SM}). By the HJW theorem, Alice can choose a local measurement on a purification of $\hat{\rho}_B$ that remotely prepares either decomposition. In ordinary quantum mechanics Bob sees the same density operator in both cases. In the strong access branch, Bob applies the forbidden discriminator/cloner and records its success flag. For the decomposition supported on the cloned set, the success probability is Eq.~(\ref{eq:Psucc0_SM}), which is strictly positive. For a generic alternative decomposition the average success probability differs, because otherwise the success functional would be affine in $\hat{\rho}_B$ and could not implement the state-dependent dependent-set cloning action. Repeating the test over many shared pairs lets Bob distinguish Alice's choice with arbitrarily high confidence. Hence the primitive is incompatible with no-signaling.
\end{proof}

\begin{theorem}[Logarithmic-query Grover consequence]
\label{thm:grover_summary_SM}
If a strong phase-access RNQM supplies unit-cost reflections $\hat{R}_{\psi}=2\ketbra{\psi}{\psi}-\id$ about the current unknown state, then unstructured search with a unique marked item can be solved with $O(\log N)$ oracle queries.
\end{theorem}

\begin{proof}---Run Protocol~3. Let $a_r=\braket{w}{\psi_r}$ and $p_r=|a_r|^2$. The recurrence in Eq.~(\ref{eq:ar_recursion_SM}) gives $a_{r+1}=(3-4p_r)a_r$ and hence $p_{r+1}=p_r(3-4p_r)^2$. While $p_r\le1/4$, the factor satisfies $(3-4p_r)^2\ge4$, so $p_{r+1}\ge4p_r$. Since $p_0=1/N$, after $r_*=\lceil\log_4(N/4)\rceil=\Theta(\log N)$ rounds one has constant success probability. A constant number of repetitions boosts the success probability. The action of $\hat{R}_{\psi_r}$ contains the overlap $\braket{\psi_r}{\phi}$, so it is an OD primitive unavailable in the protected branch.
\end{proof}

%=====================================================================================================================
\section{Summary of implications}\label{sec:summary}
%=====================================================================================================================

The protected quotient theory and the phase-access theory have different operational content:
\begin{center}
\setlength{\tabcolsep}{0.25in}
\begin{adjustbox}{max width=0.95\textwidth}
\begin{tabular}{lll}
\toprule
Question & Protected branch & Access branch \\
\midrule
What is the flag? & \makecell[l]{Gauge coordinate carrying $\hat{J}$} & \makecell[l]{Physical phase degree of freedom} \\
Allowed effects & $\Sym(V)\cap\{\hat{J}\}'$ & \makecell[l]{Some $\hat{J}$-noncommuting effects} \\
Global phase & Unobservable & Observable \\
Composition & \makecell[l]{Balanced product $\bal\simeq\otimes_{\C}$} & \makecell[l]{Ill-defined quotient or larger real theory} \\
Relation to CQM & \makecell[l]{Empirically equivalent realification} & Not CQM \\
If generic/coherent & \makecell[l]{No post-quantum primitive} & \makecell[l]{OD resource: cloning, signaling, $O(\log N)$ Grover} \\
\bottomrule
\end{tabular}
\end{adjustbox}
\end{center}

Thus, the operational content of the phase flag is decisive. If the flag is protected, the quotient construction is complex quantum mechanics in real notation. If the flag is accessible, it is not a harmless real rewriting; it is a phase-reference extension.

\end{document}